\documentclass{amsart} 
\usepackage{graphicx}
\usepackage{amssymb, amsmath, sidecap}
\usepackage{amsfonts}
\usepackage{amssymb}
\usepackage{float}
\usepackage{mathrsfs}
\usepackage{lineno}
\usepackage{empheq}
\usepackage[pdftex]{hyperref}

\newcommand\myeq{\mathrel{\stackrel{\makebox[0pt]{\mbox{\normalfont\tiny def}}}{=}}}

\newcommand{\N}{\mathbb{N}}

\newcommand{\R}{\mathbb{R}}
\newcommand{\C}{\mathbb{C}}

\def\U{U_0}
\def\tilde{\widetilde}

\renewcommand{\Re}{\operatorname{Re}}
\renewcommand{\Im}{\operatorname{Im}}

\def\bt{\begin{thm}}
\def\et{\end{thm}}
\def\bl{\begin{lem}}
\def\el{\end{lem}}
\def\bd{\begin{defi}}
\def\ed{\end{defi}}
\def\bc{\begin{cor}}
\def\ec{\end{cor}}
\def\bp{\begin{proof}}
\def\ep{\end{proof}}
\def\br{\begin{rem}}
\def\er{\end{rem}}
\newcommand{\inner}[2]{\langle #1, #2 \rangle}

\newtheorem{thm}{Theorem}[section]
\newtheorem{lem}{Lemma}[section]
\newtheorem{defi}{Definition}[section]

\newtheorem{rem}{Remark}[section]
\newtheorem{cor}{Corollary}[section]

\numberwithin{equation}{section}
\numberwithin{figure}{section}

\allowdisplaybreaks

\begin{document}
\title[Two-Layer Quasigeostrophic  Channel Model]{Baroclinic Instability and Transitions in a Two-Layer Quasigeostrophic  Channel Model}

\author[Cai]{Ming Cai}
\address[]{Earth, Ocean, \& Atmospheric Science College of Arts and Sciences\\
Florida State University\\
Tallahassee, Florida 32306-4520}
\email{cai@met.fsu.edu}

\author[Hernandez]{Marco Hernandez}
\address[]{Department of Mathematics,
Indiana University, Bloomington, IN 47405}
\email{hernmarc@indiana.edu}

\author[Ong]{Kiah Wah Ong}
\address[]{Department of Mathematics,
Indiana University, Bloomington, IN 47405, USA and Department of Mathematical and Actuarial Sciences, Universiti Tunku Abdul Rahman, Cheras, 43000 Kajang, Selangor, Malaysia}
\email{kiahong@indiana.edu}

\author[Wang]{Shouhong Wang}
\address[]{Department of Mathematics,
Indiana University, Bloomington, IN 47405}
\email{showang@indiana.edu}

\date{}

\thanks{This research is supported in part by the National Science Foundation (NSF) grant DMS-1515024, and by the Office of Naval Research (ONR) grant N00014-15-1-2662.}


\begin{abstract}
The main objective of this article is to derive a mathematical theory associated with the nonlinear stability and dynamic transitions of the basic shear flows associated with baroclinic instability, which plays a fundamental role in the dominant mechanism shaping the cyclones and anticyclones that dominate weather in mid-latitudes, as well as the mesoscale  eddies that play various roles in oceanic dynamics and the transport of tracers. 
This article provides a general transition and stability theory for the two-layer quasi-geostrophic model originally derived by Pedlosky \cite{ped1970}. 
We show that the instability and dynamic transition of the basic shear flow occur only when the Froude number  $F>(\gamma^2+1)/2$, where $\gamma$  is the weave number of lowest zonal harmonic. In this case, we derive a precise critical shear velocity $U_c$ and a dynamic transition number $b$, and we show that the dynamic transition and associated flo0w patterns are dictated by the sign of this transition number. In particular, we show that for $b>0$, the system undergoes a continuous transition, leading to spatiotemporal flow oscillatory patterns as the shear velocity $U$ crosses $U_c$; for $b<0$, the system undergoes a jump transition, leading to drastic change and more complex transition patterns in the far field, and to unstable period solutions for $U<U_c$.
Furthermore, we show that when the wavenumber of the zonal harmonic $\gamma=1$,  the transition number $b$ is always positive and  the system always undergoes a continuous dynamic transition leading spatiotemporal oscillations. This suggests that a continuous transition to spatiotemporal patterns is preferable for the shear flow associated with baroclinic instability.  
\end{abstract}

\maketitle

\section{Introduction}

Geophysical fluid flows and climate variability exhibit recurrent large-scale patterns which are directly  linked to dynamical processes represented in the governing dissipative dynamical system of the atmosphere and the ocean. The study of  the persistence of these patterns and the transitions between them are of fundamental importance in geophysics fluid dynamics and climate dynamics. 

Baroclinic instability is one of the most important 
geophysical fluid dynamical instability, and plays a crucial role in understanding  the dominant mechanism shaping the cyclones and anticyclones that dominate weather in mid-latitudes, as well as the mesoscale  eddies that play various roles in oceanic dynamics and the transport of tracers. 

Models for geophysics fluid dynamics and climate dynamics are
based on the conservation laws of fluid mechanics and  consist of
systems of  nonlinear partial differential equations (PDEs).  These
can be put into the perspective of infinite-dimensional dissipative systems
exhibiting  large-dimensional attractors. The global attractor is a mathematical object strongly connected to  the overall dissipation in the system. Weather regimes and climate variability  are, however,  often associated with  dynamic transitions between different regimes, each  represented by local attractors. 

This type of physically-induced stability and transition  leads naturally for us to search for  the full set of  transition states, represented by a local attractor, giving a complete characterization of stability and transition. Such study was initiated in early 2000 by Ma and Wang, and the corresponding dynamical transition theory and its various applications are synthesized in \cite{ptd}; see also 
\cite{b-book, MW09c, MW09a}. In particular, we have  shown that the transitions of all dissipative systems can  be  classified into three classes: continuous, catastrophic and random, which correspond to very different dynamical transition behavior of the system. Basically, as the control parameter passes the critical threshold,  the transition states stay in a close neighborhood of the basic state for a continuous (also called Type-I) transition,  are outside of a neighborhood of the basic state for a catastrophic (also called jump or Type-II) transition. For the random (also called mixed or Type-III) transition, a neighborhood is divided into two open regions with a continuous transition in one region, and a jump transition in the other region. 

The problem of nonlinear stability and transition of baroclinically unstable  waves is not only of fundamental theoretical interest, but also of practical value. The first analytic study of the problem was made by 
Eric Eady \cite{eady49}, 
Joseph Pedlosky \cite{ped1970} in the context  of a single unstable baroclinic wave, followed by the work of Mak \cite{mak1985}, Cai and Mak \cite{cai87} and  Cai  \cite{cai92}, among many others. 
In view of all the previous work for this problem and to the best of our knowledge,  the study of nonlinear stability and  transitions of the basic flows, however, is still  open.  

The main objective of this paper is to study the nonlinear dynamic transition near the onset of baroclinic linear instability. 
The model we adopt is a two-layer quasi-geostrophic model originally derived by \cite{ped1970} and used in \cite{cai87, cai92}. 
The idea of using simplified model to characterize the specific geophysical phenomena goes back at least to  the earlier workers in this field, such as J. Charney,   R.  Fj{\"o}rtoft, and J. von Neumann \cite{charney1950}; see also \cite{dijkstra2005, pedlosky2013}  and the references therein. From the lessons learned by the failure of Richardson's pioneering work, we try to be satisfied with simplified models approximating the actual motions to a greater or lesser degree instead of attempting to deal with the atmosphere/ocean in its full complexity. By starting with models incorporating only what are thought to be the most important of atmospheric influences, and by gradually bringing in others, one is able to proceed inductively and thereby to avoid the pitfalls inevitably encountered when a great many poorly understood factors are introduced all at once.

We explain now the main results obtained in this article. First, we consider the (nondimensional) basic shear flow velocities on the two layers given by 
$$U_1=-U_2 =U,$$
on a nondimensional domain $\mathcal{R}=(0,2\pi \gamma^{-1})\times(0,\pi)$ with  $\gamma>0$  being   the wavenumber of the lowest zonal harmonic.

We show that in the case where the Froude number $F$, defined by \eqref{parameters},  satisfies  $F< (\gamma^2+1)/2$, the basic shear flow is always stable. 

For the case where  $F>(\gamma^2+1)/2$, we show that the instability and dynamic transition occurs at the following critical shear velocity; see also \eqref{Uc}: 
\begin{equation}\label{Uc-1}
U_c^2 \ \myeq \ \min_{\substack{ k,l\geq 1 \\  \lambda_{k,l}<2F}} \frac{1}{2F-\lambda_{k,l}} \left( \frac{F^2\beta^2}{\lambda_{k,l}(F+\lambda_{k,l})^2} + \frac{\lambda_{k,l}(\frac{1}{Re}\lambda_{k,l}+r)^2}{\gamma^2 k^2}\right), 
\end{equation}
where $\lambda_{k,l}=\gamma^2 k^2 +l^2$. In other words, we show that as the shear velocity $U$ crosses the critical value $U_c$, the basic shear flow becomes unstable, and always undergoes a dynamic transition to one the three types as mentioned earlier. 

Then we show that the types of transitions and the detailed structure of the transition states are 
 dictated by the sign of the following nondimensional transition number:
\begin{equation}
b\ \myeq \ (2F-\lambda_{\hat{k},\hat{l}})(\gamma^2\hat{k}^2-\hat{l}^2)+2\hat{l}^2\lambda_{\hat{k},\hat{l}},  
\end{equation}
where $\hat{k} \ge 1$  and $ \hat{l}\ge 1$ are the wave numbers that   minimize in \eqref{Uc-1}.

More precisely, if  the transition number  $b>0$,  as $U$ crosses the critical value $U_c$, the underlying system undergoes a continuous transition leading to periodic solutions (oscillatory mode). The spatiotemporal patterns depicted by the oscillatory mode 
plays an important role in the study of climate variabilities. 

If   the transition number $b<0$, the system undergoes a drastic change (jump transition) at the critical shear velocity $U_c$. Meanwhile, the system bifurcates to an unstable period solution for $U<U_c$. In this case,  as $U$ increasingly crosses $U_c$,  the physical transition states, represented by a local attractor, display more complex structure.

We remark  that when the wavenumber of the zonal harmonic $\gamma=1$,  the transition number $b$ is always positive, so that  the system always undergoes a continuous dynamic transition leading spatiotemporal oscillations described by  \eqref{approx_soln}. This suggests that a continuous transition to spatiotemporal patterns is favorable for the shear flow associated with baroclinic instability.  

The key mathematical ingredient of the study in this article is to fully characterize the interactions of different modes (low-low, low-high and high-high), and their effects on the transition types and transitions states.  The central gravity of the analysis consists of the following: 

\begin{itemize}
\item full description of both stable and unstable modes for the linearized problem formulated as the principle of exchange of stabilities (PES),
\item  deriving the leading-order approximations of the slaving of the stable modes in terms of the unstable modes,
\item reducing the full governing  partial differential equations to the center manifold generated by the unstable modes, and 
\item characterizing the dynamical transitions and its physical implications. 
\end{itemize}

The paper is organized as follows. In Section 2 we briefly introduce the two-layer quasigeostrophic model for baroclinic instability based on \cite{ped1970, cai87}. The main theorem of the paper is given in Section 3.
For completeness, the well-posedness of the model is presented in Section 4; our arguments rely heavily on the spectral properties of the linearized system, the classical theory of semigroups, and some appropriate a priori estimates. Section 5 establishes the principle of exchange of stabilities (PES).
Section 6 provides  the proof of the main theorem, Theorem \ref{transition}.

\section{A Two-Layer Model for Baroclinic Instability}

We start with the basic two-layer model proposed  in \cite{ped1970}, in which two layers of  homogeneous fluids, each with a different but constant density, lie on a horizontal plane rotating with angular velocity $\Omega$. 

Then the two-layer quasi-geostrophic model is derived as the leading order equations with respect to the Rossby number 
\begin{equation}
\varepsilon={\U}/(f_0 L), \label{rossby}
\end{equation}
where $\U$ is the typical horizontal velocity, and $f_0$  is the leading term in  the classical Coriolis expansion $2\Omega = f_0 + \beta' y'$, $y'=Ly$, and $L$ is the typical horizontal scale. 
The model obtained by the above procedure is  given by (2.13) in \cite{ped1970}, and reproduced below
\begin{align}
&\left[\frac{\partial}{\partial t} + \frac{\partial \psi_1}{\partial x} \frac{\partial}{\partial y}- \frac{\partial \psi_1}{\partial y} \frac{\partial}{\partial x}\right] 
\left[\Delta \psi_1 + F (\psi_2-\psi_1) + \beta y\right] = -r \Delta \psi_1, \\
&\left[\frac{\partial}{\partial t} + \frac{\partial \psi_2}{\partial x} \frac{\partial}{\partial y}- \frac{\partial \psi_2}{\partial y} \frac{\partial}{\partial x}\right] 
\left[\Delta \psi_2 + F (\psi_1-\psi_2) + \beta y\right] = -r \Delta \psi_2, 
\end{align}
where $\psi_1$ and $\psi_2$ are the geostrophic stream-functions for the horizontal motion, and the nondimensional parameters are given as follows:

\begin{equation}
\label{parameters}
\begin{aligned}
&E=\frac{2\nu}{f_0 D^2} && \text{the Ekman number,}\\
&\beta=\frac{\beta' L^2}{\U} && \text{the planetary vorticity factor,}\\
&\frac{\Delta \rho}{\rho} = \frac{\rho_2-\rho_1}{\rho_2} && \text{the density ratio,}\\
&F=\frac{f_0^2 L^2}{\frac{\Delta\rho}{\rho} g \frac{D}{2}} && \text{the internal rotational Froude number,}\\
&\delta=\frac{D}{L} && \text{the aspect ratio.}
\end{aligned}
\end{equation}

In the derivation of the above equations the ratio $\frac{E}{\epsilon}$ has been neglected, as is customary when boundary layer effects are not considered. However, for mathematical purposes, it is convenient to keep the associated viscous terms. The equations then become

\begin{align}
&\left[\frac{\partial}{\partial t} + \frac{\partial \psi_1}{\partial x} \frac{\partial}{\partial y}- \frac{\partial \psi_1}{\partial y} \frac{\partial}{\partial x}\right] 
\left[\Delta \psi_1 + F (\psi_2-\psi_1) + \beta y\right] = -r \Delta \psi_1+\frac{1}{Re}\Delta^2 \psi_1, \\
&\left[\frac{\partial}{\partial t} + \frac{\partial \psi_2}{\partial x} \frac{\partial}{\partial y}- \frac{\partial \psi_2}{\partial y} \frac{\partial}{\partial x}\right] 
\left[\Delta \psi_2 + F (\psi_1-\psi_2) + \beta y\right] = -r \Delta \psi_2+\frac{1}{Re}\Delta^2\psi_2, 
\end{align}
where $Re$ is the Reynolds number,
\begin{equation}
\frac{2}{Re}= \frac{2 \nu L}{\U D} = \frac{E}{\epsilon}.
\end{equation}

Our analysis requires this quantity to be strictly positive, but the final results are qualitatively independent of its value as long as it is indeed a small number.

The basic (shear-type) flow is given by
\begin{align*}
\psi^{(0)}_1 &= -U_1 y, \\
\psi^{(0)}_2 &= -U_2 y.
\end{align*}

Now set
\begin{equation}
\varphi_i = \psi_i - \psi^{(0)}_i \qquad i=1,\,2.
\end{equation}

The equations for the deviation stream-functions $\varphi_1$ and $\varphi_2$ are given by
\begin{align}
\left[ \frac{\partial}{\partial t} + U_1 \frac{\partial}{\partial x} + \frac{\partial \psi_1}{\partial x} \frac{\partial}{\partial y}- \frac{\partial \psi_1}{\partial y} \frac{\partial}{\partial x} \right]\left[\Delta\varphi_1+F(\varphi_2-\varphi_1)+(U_1-U_2)y+\beta y\right] \\
\nonumber =-r\Delta \varphi_1+\frac{1}{Re}\Delta^2\varphi_1, \\
\left[ \frac{\partial}{\partial t} + U_2 \frac{\partial}{\partial x} + \frac{\partial \psi_1}{\partial x} \frac{\partial}{\partial y}- \frac{\partial \psi_1}{\partial y} \frac{\partial}{\partial x} \right]\left[\Delta\varphi_2+F(\varphi_1-\varphi_2)+(U_2-U_1)y+\beta y\right] \\
\nonumber =-r\Delta \varphi_2+\frac{1}{Re}\Delta^2\varphi_2.
\end{align}

The case considered by Mak \cite{mak1985}, Cai and Mak \cite{cai87} is that a westerly mean shear, $2U$, with an associated meridional thermal contract, is imposed as a potential source of energy to sustain the wave field, and thus counteracting dissipation. In the above framework this amounts to setting
\begin{equation}
U_1=U\quad\text{and}\quad U_2=-U.
\end{equation}

In this case, if we let
\begin{align}
& \psi=\frac{1}{2}\left(\varphi_1+\varphi_2\right)=\frac{1}{2}\left(\psi_1+\psi_2\right), \\
& \theta=\frac{1}{2}\left(\varphi_1-\varphi_2\right)=\frac{1}{2}\left(\psi_1-\psi_2\right)+Uy,
\end{align}
then the resulting equations for $\psi$ and $\theta$ are
\begin{equation}\label{main1}
\begin{aligned}
& \frac{\partial}{\partial t}\Delta\psi + J(\psi,\Delta\psi+\beta y) + J(\theta-Uy,\Delta\theta) = -r\Delta\psi+\frac{1}{Re}\Delta^2\psi, \\
& \frac{\partial}{\partial t}\left(\Delta\theta-2F\theta\right)+2FJ(\theta-Uy,\psi)+J(\psi,\Delta\theta) + J(\theta-Uy,\Delta\psi+\beta y)\\
& \qquad \qquad  = -r\Delta\theta+\frac{1}{Re}\Delta^2\theta, 
\end{aligned}
\end{equation}
where we used the notation $J(u,v)=\frac{\partial u}{\partial x}\frac{\partial v}{\partial y}-\frac{\partial u}{\partial y}\frac{\partial v}{\partial x}$.

This system is the one studied throughout the rest of the article. The non-dimensional domain is taken to be the rectangle 
\begin{equation} \label{domain}
\mathcal{R}=(0,2\pi \gamma^{-1})\times(0,\pi),
\end{equation}
where $\gamma>0$  is  the wavenumber of the lowest zonal harmonic.

The system is supplemented with periodic boundary conditions on the horizontal axis, 
\begin{equation}\label{per}
\psi(x,y,t)=\psi\left(x+\frac{2\pi}{\gamma},y,t\right),\ \ \  \theta(x,y,t)=\theta\left(x+\frac{2\pi}{\gamma},y,t\right),
\end{equation}
whereas, on the upper and lower segments of the boundary, we impose 
\begin{align}
&  
\label{bdry}
\begin{aligned}
&\psi(x,0,t)=0, \qquad && \theta(x,0,t) = 0, \\ 
&\psi(x,\pi,t)=0, \quad && \theta(x,\pi,t) = 0,
\end{aligned}
\\
&  \label{bdry2}
\begin{aligned}
&\Delta \psi(x,0,t)=0, \qquad &&\Delta \theta(x,0,t) = 0, \\ 
&\Delta \psi(x,\pi,t)=0, \qquad &&\Delta \theta(x,\pi,t) = 0.
\end{aligned}
\end{align}

\section{Mathematical Formulation and Main Theorems}

For every $v\in L^2(\mathcal{R})$ we let $u=R_1 v$ be the unique solution of the Dirichlet problem
\[
\begin{aligned}
& -\Delta u = v &&  \text{in } \mathcal{R},\\
& u=0  &&  \text{on } \partial\mathcal{R}.
\end{aligned}
\]
Note that $R_1$ is a bounded linear map from $L^2(\mathcal{R})$ into $H^2(\mathcal{R})\cap H^1_0(\mathcal{R})$.

Similarly, we let $u=R_2 v$ be the unique solution of
$$
\begin{aligned}
& -\Delta u +2F u= v  &&  \text{in } \mathcal{R},\\
& u=0  &&  \text{on } \partial\mathcal{R}.
\end{aligned}
$$

Note that since $F\geq 0$ the operator $R_2$ is well defined and, as $R_1$, it maps $L^2(\mathcal{R})$ continuously into $H^2(\mathcal{R})\cap H^1_0(\mathcal{R})$.

By making use of the boundary conditions we see that $R_1(-\Delta)u=u$ and $R_2(-\Delta u+2Fu)=u$, for any smooth $u$ satisfying \eqref{per} and \eqref{bdry}. Furthermore, if $u$ also satisfies \eqref{bdry2}, then one has $R_1 \Delta^2 u = -\Delta u$, and even
\[
R_2 \Delta^2 u = -\Delta u -2F u+4F^2 R_2 u.
\]

Thus, applying $R_1$ and $R_2$ to equations \eqref{main1},  we obtain
\begin{align*}
&
\frac{\partial\psi}{\partial t} = -r\psi +\frac{1}{Re}\Delta\psi +R_1\left( J(\psi,\Delta\psi+\beta y)+J(\theta-Uy,\Delta\theta)\right), 
\\
& 
\frac{\partial\theta}{\partial t} = -r\theta +2rF R_2\theta+\frac{1}{Re}\left( \Delta\theta+2F \theta -4F^2 R_2 \theta\right)\\
&\qquad \qquad +R_2\left(J(-\Delta\theta+2F(\theta-Uy),\psi)+J(\theta-Uy,\Delta\psi+\beta y)\right).
\end{align*}

The above suggests that we introduce the linear operator $L=-A+B$, where
\[
A(\psi,\theta)=
\begin{pmatrix} 
r\psi-\frac{1}{Re}\Delta\psi \\
r\theta-2rFR_2\theta-\frac{1}{Re} \left( \Delta\theta+2F \theta -4F^2 R_2 \theta\right)
\end{pmatrix},  
\]

\[
B(\psi,\theta)=
\begin{pmatrix} 
R_1(\beta J(\psi,y)+UJ(\Delta\theta,y)) \\
R_2(2FU J(\psi,y)+\beta J(\theta,y)+UJ(\Delta\psi,y))
\end{pmatrix},  
\]
and the bilinear operator
\[
G\left((\psi,\theta),(\tilde{\psi},\tilde{\theta})\right)=
\begin{pmatrix} 
R_1(J(\psi,\Delta\tilde{\psi})+J(\theta,\Delta\tilde{\theta})) \\
R_2(J(-\Delta\theta+2F\theta,\tilde{\psi})+J(\theta,\Delta\tilde{\psi}))
\end{pmatrix}.
\]

With these definitions the initial value problem associated with \eqref{main1} can now be written in the more compact form
\begin{equation}\label{abst}
\begin{aligned}
& \dfrac{du}{dt} = L u + G(u,u),\\
& u(0)=u_0=(\psi_0,\theta_0).
\end{aligned}
\end{equation}

In order to study the above evolution equation we introduce the Hilbert spaces
\begin{align*}
& V=\left\lbrace u=(\psi,\theta) \in H^3(\mathcal{R})^2:\: \psi\text{ and } \theta \text{ satisfy \eqref{per}, \eqref{bdry} and \eqref{bdry}}\right\rbrace,
\\
& 
H=\left\lbrace u=(\psi,\theta) \in H^2(\mathcal{R})^2:\: \psi\text{ and } \theta \text{ satisfy \eqref{per}, \eqref{bdry}}\right\rbrace.
\end{align*}

Our first result, pertaining to the well-posedness of \eqref{main1}, can now be formulated as follows.

\begin{thm}\label{existence}
For every  $T>0$ and $\psi_0,\,\theta_0 \in H^2(\mathcal{R})$ satisfying \eqref{per} and \eqref{bdry} there exists a unique solution $u=(\psi,\theta)$ of \eqref{abst} in $C([0,T];H)\cap L^2(0,T; V)$.
\end{thm}

Next we study the dynamic stability and transition  behavior of the zonal flow, represented here by the trivial solution $\psi=\theta\equiv 0$. For this purpose, we introduce the following critical shear, defined for $2F>\gamma^2+1$ by
\begin{equation}\label{Uc}
U_c^2 =\min_{\substack{ k,l\geq 1 \\  \lambda_{k,l}<2F}} \mathcal{U}^2(k,l), 
\end{equation}
where 
\begin{align}
& \lambda_{k,l}=\gamma^2 k^2 +l^2,  \label{lambda}\\
& \mathcal{U}^2(k,l)=\frac{1}{2F-\lambda_{k,l}} \left( \frac{F^2\beta^2}{\lambda_{k,l}(F+\lambda_{k,l})^2} + \frac{\lambda_{k,l}(\frac{1}{Re}\lambda_{k,l}+r)^2}{\gamma^2 k^2}\right). \label{u-defined}
\end{align}
Here we recall that the $\gamma>0$  is  the wavenumber of the lowest zonal harmonic. 

Hereafter we always assume that the minimum in \eqref{Uc} is attained at a unique pair $(\hat{k},\hat{l})$  ($\hat{k} \ge 1, \hat{l}\ge 1$):
\begin{equation}\label{Uc-attained}
U_c^2 = \mathcal{U}^2(\hat{k},\hat{l}).
\end{equation}
It is clear that this assumption is generically true. 

 Next we define a bifurcation number:
\begin{equation}\label{cont_condition}
b=(2F-\lambda_{\hat{k},\hat{l}})(\gamma^2\hat{k}^2-\hat{l}^2)+2\hat{l}^2\lambda_{\hat{k},\hat{l}}.
\end{equation}
The  main result of the paper is then the following.

\begin{thm}\label{transition}

\begin{enumerate}

\item Let $2F\leq \gamma^2+1$. Then for system  \eqref{main1}, 
the equilibrium $(\psi,\theta)\equiv (0,0)$ is a stable solution to \eqref{main1} for all values of $U$.

\item Let $2F>\gamma^2+1$. If $b> 0$ with $b$ defined by (\ref{cont_condition}), then \eqref{main1} undergoes a continuous transition to a stable periodic orbit as $U$ crosses $U_c$. More precisely, the trivial solution $(\psi,\theta)\equiv (0,0)$ is stable for $U\leq U_c$, and for $U>U_c$ it bifurcates to a stable periodic orbit, approximately given by
\begin{equation}\label{approx_soln}
\begin{aligned}
\psi(x,y,t) &= \rho \gamma \hat{k} U_c \sin\left( \omega t+ \gamma \hat{k}x\right)\sin\hat{l}y+O(|U-U_c|),\\
\theta(x,y,t)&=\rho \left(\frac{1}{Re}\lambda_{\hat{k},\hat{l}}+r\right)\cos\left(\omega t+\gamma \hat{k}x\right)\sin\hat{l}y\\
&+\rho \dfrac{\gamma\hat{k}\beta F}{\lambda_{\hat{k},\hat{l}}(F+\lambda_{\hat{k},\hat{l}})}\sin\left(\omega t+\gamma \hat{k}x\right)\sin\hat{l}y+O(|U-U_c|),
\end{aligned}
\end{equation}
where
\begin{equation}\label{rho omega thm}
\begin{aligned}
\rho^2&=\frac{4(2F-\lambda_{\hat{k},\hat{l}})(\frac{1}{Re}\lambda_{0,2\hat{l}}+r)(U-U_c)}{ \gamma^2 \hat{k}^2 F[  (2F-\lambda_{\hat{k},\hat{l}})(\gamma^2\hat{k}^2-\hat{l}^2)+2\hat{l}^2\lambda_{\hat{k},\hat{l}}] (\frac{1}{Re}\lambda_{\hat{k},\hat{l}}+r)}, \\
\omega&=\frac{\gamma \hat{k} \beta}{F+\lambda_{\hat{k},\hat{l}}}+O\left(|U-U_c|^{{3}/{2}}\right).
\end{aligned}
\end{equation}

\item Let $2F>\gamma^2+1$. If $b< 0$, then  \eqref{main1} undergoes a jump transition 
 as $U$ crosses $U_c$. Also, the system bifurcates to an unstable periodic solution of the form similar to (\ref{approx_soln}) for $U<U_c$. 

\end{enumerate}
\end{thm}

Two remarks are now in order.

\begin{rem}
By neglecting the ratio $\frac{E}{\epsilon}=\frac{2}{Re}$ in the above equations we obtain the slightly simpler formula for $\rho$
\begin{equation*}
\rho^2=\frac{4(2F-\lambda_{\hat{k},\hat{l}})(U-U_c)}{ \gamma^2 \hat{k}^2 F[  (2F-\lambda_{\hat{k},\hat{l}}) (\gamma^2\hat{k}^2-\hat{l}^2)+2\hat{l}^2\lambda_{\hat{k},\hat{l}}] }.
\end{equation*}
\end{rem}

\begin{rem} When the nondimensional wavenumber of the zonal harmonic $\gamma=1$, it is easy to see that the transition number $b>0$. Consequently, Assertion (2) of Theorem~\ref{transition}  is always true. Namely, for $\gamma=1$, the system always undergoes a continuous dynamic transition leading spatiotemporal oscillations described by  \eqref{approx_soln}. 
\end{rem}

\section{Well-Posedness of the Initial Boundary Value Problem}

We begin with some basic properties of the linear operator $L$ appearing in \eqref{abst}.

First of all, note that all the trigonometric polynomials satisfying \eqref{per}, \eqref{bdry} and \eqref{bdry2} are of the form
\begin{equation}\label{trig}
\begin{aligned}
& p(x,y)=\sum_{k\in\mathbb{Z}} \sum_{l\geq 1} p_{k,l}\varphi_{k,l}(x,y),\\
&  \varphi_{k,l}(x,y)=e^{i\gamma k x} \sin ly,\qquad 
\overline{p_{k,l}}&=p_{-k,l}, \qquad \forall k\in\mathbb{Z},\, l\ge 1,
\end{aligned}
\end{equation}
where only finitely many of the coefficients $p_{k,l}\in \mathbb{C}$ are non-zero.

Besides providing a founding step for the forthcoming analysis, the fundamental system $\{\varphi_{k,l}\}$ will play a major role in the subsequent sections where we study the local behavior of \eqref{main1} near the origin.

Next we introduce the set
\[ \mathcal{V} = \left\lbrace (p,q)\in C^\infty(\mathcal{R})^2:\: p,\, q \text{ are of the form \eqref{trig}}\right\rbrace.\]

Note that $\mathcal{V}$ is dense in both $H$ and $V$ for their respective norms. To be more precise, it is clear that the functions $\varphi_{k,l}$ satisfy $ -\Delta \varphi_{k,l} = \lambda_{k,l}\varphi_{k,l}$, where we recall that $\lambda_{k,l}=\gamma^2k^2+l^2$. 
Furthermore, the inner product of $H$ evaluated on elements $(p,q)$ and $(p',q')$ in $\mathcal{V}$ is given by
\begin{equation}\label{inner H}
\inner{(p,q)}{(p',q')}_H = \sum_{k,l} \lambda_{k,l}^2 \left( p_{k,l}\overline{p'_{k,l}}+q_{k,l}\overline{q'_{k,l}}\right). 
\end{equation}

In addition, the inner product of $V$ evaluated on elements $(p,q)$ and $(p',q')$ in $\mathcal{V}$ is equivalent to
\begin{equation}\label{inner V}
\inner{(p,q)}{(p',q')}_V = \sum_{k,l} \lambda_{k,l}^3 \left( p_{k,l}\overline{p'_{k,l}}+q_{k,l}\overline{q'_{k,l}}\right). 
\end{equation}

One can also give a detailed description of the operator $L=-A+B$ introduced in the previous section using this fundamental system. Indeed, for $p,q$ of the form \eqref{trig} one has
\begin{equation}\label{A and B op}
\begin{aligned}
A^1_{k,l}(p,q)&=\left(r +\frac{1}{Re}\lambda_{k,l}\right)p_{k,l},\\
A^2_{k,l}(p,q)&=\frac{\lambda_{k,l}}{2F+\lambda_{k,l}}\left(r +\frac{1}{Re}\lambda_{k,l} \right)q_{k,l}, \\
B^1_{k,l}(p,q)&=-i\frac{\gamma k\beta}{\lambda_{k,l}}p_{k,l}+i\gamma k U q_{k,l},\\
B^2_{k,l}(p,q)&=-i\frac{\gamma k\beta}{2F+\lambda_{k,l}}q_{k,l}+i\gamma k U\frac{2F-\lambda_{k,l}}{2F+\lambda_{k,l}} p_{k,l},
\end{aligned}
\end{equation}
where, for a generic linear operator $M$ on $\mathcal{V}$ we use the notation
\[M u= \left( \sum_{k,l}   M^1_{k,l} u\, \varphi_{k,l},\, \sum_{k,l} M^2_{k,l}u\, \varphi_{k,l}\right).\]

In virtue of the above considerations we see that
\[
\begin{aligned}
\inner{Au}{u'}_{H} &= \sum_{k,l} \lambda_{k,l}^2\left(r +\frac{1}{Re}\lambda_{k,l}\right)\left(p_{k,l}\overline{p'_{k,l}}+\frac{\lambda_{k,l}^2}{(2F+\lambda_{k,l})^2} q_{k,l}\overline{q'_{k,l}}\right),\\
\| Bu\|_{H}^2 &= \sum_{k,l} \gamma^2 k^2 \left( \Big| \beta p_{k,l} - \lambda_{k,l} U q_{k,l}\Big|^2\right. \\
&\left.+\frac{\lambda^2_{k,l}}{(2F+\lambda_{k,l})^2}\Big|\beta q_{k,l}-(2F-\lambda_{k,l})U p_{k,l}\Big|^2\right),
\end{aligned} 
\]
whence, by the density of $\mathcal{V}$ in $H$ and $V$, we get for all $u,v\in V$ the estimates
\[
\begin{aligned}
&\inner{Au}{u}_{H}\geq c\left( r \|u\|^2_H+\frac{1}{Re}\|u\|^2_V\right),\\
 &|\inner{Au}{v}_H| \leq r \|u\|_H\|v\|_H+\frac{1}{Re} \|u\|_V\|v\|_V,\\
&\| Bu\|_{H}^2  \leq C \left( \beta^2 \|u\|_{H}^2+ U^2\|u\|_{V}^2\right).
\end{aligned} 
\]

With these constructions in place we are now in position to address the global well-posedness of \eqref{main1}.
\begin{proof}[Proof of Theorem \ref{existence}]
We begin with local existence of strong solutions. More precisely, we consider the mild formulation of \eqref{abst}, which takes the form
\[
u(t) = e^{tL} u_0 + \int_0^t e^{(t-s)L} G(u(s),u(s)) ds.
\]

Recall that $L=-A+B$, where we have shown that $A$ is a self-adjoint positive operator on $H$, whose domain is clearly $D(A)=V\cap H^4(\mathcal{R})^2$, and $B$ is bounded from $V$ into $H$. In fact, one can even show that $V=D(A^{\frac{1}{2}})$. Thus, since the inclusions $D(A)\subset V \subset H$ are all dense and compact, it follows that $L$ is a sectorial operator with compact resolvent and domain $D(L)=D(A)$.

Regarding the non-linear operator $G$, we write $G(u,u)=(G^1(u),\, G^2(u))$, where
\[ 
\begin{aligned}
G^1((p,q))&=R_1(J(p,\Delta p)+J(q,\Delta q)), \\
G^2((p,q))&=R_2(J(p,\Delta q-2F q)+J(q,\Delta p)).
\end{aligned}
\]

Taking the $H^1_0$ inner product of $G^1((p,q))$ with an arbitrary $\xi\in H^2\cap H^1_0$ we get, using the definition of $R_1$,
\[
\inner{ G^1((p,q))}{\xi}_{H^1_0}=\inner{J(p,\Delta p)+J(q,\Delta q)}{\xi}_{L^2}.
\]

Similarly, using the definition of $R_2$ we see that
\[
\begin{aligned}
\inner{ G^2((p,q))}{\xi}_{H^1_0}=&\inner{J(p,\Delta q-2F q)+J(q,\Delta p)}{\xi}_{L^2}\\
&-2F\inner{R_2(J(p,\Delta q-2F q)+J(q,\Delta p))}{\xi}_{L^2}.
\end{aligned}
\]

Now, using the boundary conditions, it is easy to see that
\[
\inner{J(p,\Delta q)}{\xi}_{L^2} = -\int_\mathcal{R} \Delta q \nabla p \cdot\nabla\xi. 
\]
Then, using the Gagliardo-Nirenberg inequality, we get
\[
|\inner{J(p,\Delta q)}{\xi}_{L^2}|\lesssim \|q\|_{H^3}\|p\|_{H^3}\|\xi\|_{H^1_0}.
\]

Arguing in a similar way for the other cases we conclude that
\[
|\inner{G(u,u)}{v}_H|\lesssim \|u\|^2_V \|v\|_H,\qquad \forall u\in V,\, v\in H.
\]

It follows,  by the standard theory of analytic semigroups (see e.g. \cite{L12}), that \eqref{abst} has a unique local mild solution. Furthermore, the maximal interval of existence coincides with the blow-up time of the $H$ norm of the solution. Thus, in order to obtain global existence, it suffices to get a priori estimates on the $H$ norm of the solution.

Let $u=(\psi,\theta)$ be the unique local mild solution of \eqref{abst}. By testing \eqref{main1} against $\psi$ and $\theta$, and using the cancellation properties of $J$, we arrive at
\begin{align*}
& \frac{1}{2}\frac{d}{dt} \|\psi\|^2_{H^1_0}+r\|\psi\|^2_{H^1_0}+\frac{1}{Re} \|\Delta\psi\|^2_{L^2} = -U\inner{\Delta \theta}{\partial_x\psi}_{L^2}+\inner{J(\psi,\theta)}{\Delta\theta}_{L^2}, \\
& \frac{1}{2}\frac{d}{dt} \left(\|\theta\|^2_{H^1_0}+2F\|\theta\|^2_{L^2}\right)+r\|\theta\|^2_{H^1_0}+\frac{1}{Re} \|\Delta\theta\|^2_{L^2} \\
& \qquad = U\inner{\Delta \theta}{\partial_x\psi}
 -\inner{J(\psi,\theta)}{\Delta\theta}_{L^2}+2FU\inner{\partial_x\psi}{\theta}_{L^2}.
\end{align*}

Adding these equations,  the terms involving $J$ cancel out and, consequently,   we deduce that 
\[ 
\frac{1}{2}\frac{d}{dt}\left( \|\psi\|^2_{H^1_0}+\|\theta\|^2_{H^1_0}+2F\|\theta\|^2_{L^2} \right)+\frac{1}{Re}\left(\|\Delta\psi\|^2_{L^2}+\|\Delta\theta\|^2_{L^2}\right) \leq 2FU \|\psi\|_{H^1_0}\|\theta\|_{L^2}. 
\]
By the Gronwall's inequality, we  can then show that 
there exist constants $c_1,c_2>0$ such that
\[ \|\psi(t)\|^2_{H^1_0}+\|\theta(t)\|^2_{H^1_0} + \int_0^t\left(\|\Delta\psi\|^2_{L^2}+\|\Delta\theta\|^2_{L^2}\right) ds \leq c_1 e^{c_2 t} \left(\|\psi_0\|^2_{H^1_0}+\|\theta_0\|^2_{H^1_0}\right).\]

Next we test \eqref{main1} against $-\Delta \psi$ and $-\Delta\theta$ to get
\begin{align*}
& \frac{1}{2}\frac{d}{dt} \|\Delta\psi\|^2_{L^2}+r\|\Delta\psi\|^2_{L^2}+\frac{1}{Re}
       \|\Delta\psi\|^2_{H^1_0} \\
&\qquad  =-\beta\inner{\partial_x\psi}{\Delta\psi}_{L^2} 
   + U\inner{\Delta \theta}{\partial_x\Delta\psi}_{L^2}-\inner{J(\theta,\Delta\theta)}{\Delta\psi}, \\
& \frac{1}{2}\frac{d}{dt} \left(\|\Delta\theta|^2_{L^2}+2F\|\theta\|^2_{H^1_0}\right) 
   +r\|\Delta\theta\|^2_{L^2}+\frac{1}{Re} \|\Delta\theta\|^2_{H^1_0} \\
& \qquad =-\beta\inner{\partial_x\theta}{\Delta\theta}_{L^2}
-U\inner{\Delta \theta}{\partial_x\Delta\psi}_{L^2}\\
& \qquad \qquad +\inner{J(\theta,\Delta\theta)}{\Delta\psi}_{L^2}-2F\inner{J(\theta,\Delta\psi)+U\partial_x\psi}{\Delta\theta}_{L^2}.
\end{align*}
Hence 
\[
\frac{1}{2}\frac{d}{dt} \left(\|u\|^2_H+2F\|\theta\|^2_{H^1_0}\right)+\frac{1}{Re} \|u\|^2_V \leq (\beta+2F)\|u\|_H\|u\|_V-2F\inner{J(\theta,\Delta\psi)}{\Delta\theta}_{L^2}.
\]

The last term above can be bounded as
\[
|\inner{J(\theta,\Delta\psi)}{\Delta\theta}_{L^2}|\lesssim \|\theta\|_{H^3} \|\theta\|_{H^2} \|\psi\|_{H^2} \leq \epsilon \|u\|^2_V+C_\epsilon \|u\|^4_H, 
\]
which shows, upon substitution in the previous estimate and applying Gronwall's inequality, that
\[ \|u(t)\|^2_H +\int_0^t \|u\|^2_V ds \leq C \exp\left( k t + \int_0^t \|u\|^2_H ds \right) \|u_0\|^2_H. \]

Since the integral on the right hand side is known to be controlled by $\|u_0\|^2_H$ by the preceding estimate, the above provides the a priori estimate required to establish  the global existence of the solutions. The proof is then complete.
\end{proof}

\section{Linear Analysis and the Principle of Exchange of Stability}

To study the nonlinear dynamic stability and transitions of the underlying system, we need a clear understanding on the eigenvalue problem for the linearized problem around the basic flow. 
The linearized evolution is governed by the operator $L$ introduced in the preceding section, which is, essentially, a perturbation of $-A$; see \eqref{A and B op}. In fact, since the spectrum and eigenfunctions of the latter can be found explicitly, it is expected that in terms of this basis the spectral properties of $L$ can be described more or less explicitly. In what follows we make these ideas precise.

The eigenfunctions of $A$ are given by
\begin{equation}\label{ei_kl}
\begin{aligned}
& e^{(1)}_{k,l} (x,y) = ( \sin(\gamma kx)\sin(ly), 0 ),\quad & e^{(2)}_{k,l} (x,y) = ( \cos(\gamma kx)\sin(ly), 0 ), \\
& e^{(3)}_{k,l} (x,y) = ( 0, \sin(\gamma kx)\sin(ly) ),\quad & e^{(4)}_{k,l} (x,y) = ( 0, \cos(\gamma kx)\sin(ly) ).
\end{aligned}
\end{equation}
where $k\geq 0$ and $l\geq 1$ are integers.

Now, although $L$ does not preserve the corresponding eigenspaces, it does have invariant spaces that can be described entirely using \eqref{ei_kl}. Indeed, for every pair $(k,l)$, the set $H_{k,l}\subset D(A)$ given by
\begin{equation}\label{Hkl}
H_{k,l}=\text{span }\left\lbrace e^{(1)}_{k,l},e^{(2)}_{k,l},e^{(3)}_{k,l},e^{(4)}_{k,l}\right\rbrace
\end{equation}
is invariant under $L$.

Moreover, the direct sum $\bigoplus_{k\geq 0,\, l\geq 1} H_{k,l}$ is dense in both $D(A)$ and $H$, for the respective norms. Therefore, the spectral properties of $L$ are completely determined by its action on this family of finite dimensional spaces. This fact considerably simplifies the forthcoming analysis, and will be used extensively in the remainder of this article.

Our first result is a complete quantitative description of the spectrum of $L$.

\begin{lem}

\begin{enumerate}
\item The spectrum of $L$ is given by $\sigma(L)=\{\mu^\pm_{k,l}\}_{k\geq 0, l\geq 1}$, where
\begin{equation}\label{mupm}
\begin{aligned}
\mu^\pm_{k,l}&=\dfrac{-(F+\lambda_{k,l})\eta_{k,l} \pm \sqrt{F^2\eta_{k,l}^2+\gamma^2 k^2U^2 (4F^2-\lambda_{k,l}^2)}}{2F+\lambda_{k,l}},\\
\eta_{k,l}&=\frac{1}{Re} \lambda_{k,l}+r+i\frac{\gamma k\beta}{\lambda_{k,l}}.
\end{aligned}
\end{equation}

\item To every integer pair $(k,l)$, with $k\geq 0$ and $l\geq 1$, there correspond two linearly independent subspaces $H_{k,l,+}$ and $H_{k,l,-}$, invariant under $L$, given by
\begin{equation}\label{Hkl+-}
H_{k,l,\pm} = \text{span }\left\lbrace v^{1,\pm}_{k,l} , v^{2,\pm}_{k,l} \right\rbrace,
\end{equation}
where
\begin{equation}\label{Hc}
\begin{aligned}
v^{1,\pm}_{k,l} & =\gamma kU e^{(1)}_{k,l}+ \left[ \Re(\mu^\pm_{k,l}+\eta_{k,l}) e^{(4)}_{k,l}+\Im(\mu^\pm_{k,l}+\eta_{k,l}) e^{(3)}_{k,l}\right], \\
v^{2,\pm}_{k,l} & =\gamma kUe^{(2)}_{k,l}+ \left[ \Im(\mu^\pm_{k,l}+\eta_{k,l}) e^{(4)}_{k,l}-\Re(\mu^\pm_{k,l}+\eta_{k,l}) e^{(3)}_{k,l}\right].
\end{aligned}
\end{equation}

\item The functions given by \eqref{Hc} satisfy the identities
\begin{equation}\label{LHkl}
Lv^{1,\pm}_{k,l} = \alpha^\pm_{k,l} v^{1,\pm}_{k,l}-\sigma^+_{k,l} v^{2,\pm}_{k,l},\quad Lv^{2,\pm}_{k,l} = \alpha^\pm_{k,l} v^{2,\pm}_{k,l}+\sigma^\pm_{k,l} v^{1,\pm}_{k,l}
\end{equation}
where
\begin{equation*}
\alpha^\pm_{k,l} = \Re \mu^\pm_{k,l} ,\qquad \sigma^\pm_{k,l}=\Im \mu^\pm_{k,l}.
\end{equation*}
\end{enumerate}
\end{lem}

\begin{proof}
Before proceeding further, it is convenient to introduce the complex valued functions
\begin{equation}\label{wi_kl}
\begin{aligned}
&w^{(1)}_{k,l}(x,y)=(e^{-i\gamma kx}\sin ly, 0)=e^{(2)}_{k,l}(x,y)-i e^{(1)}_{k,l}(x,y),\\
&w^{(2)}_{k,l}(x,y)=(0, e^{-i\gamma kx}\sin ly)=e^{(4)}_{k,l}(x,y)-i e^{(3)}_{k,l}(x,y).
\end{aligned}
\end{equation}

These functions form a basis for the complexification of $H_{k,l}$, and the action of $L$ in these coordinates takes the simple form
\begin{equation}\label{Lw}
\begin{aligned}
Lw^{(1)}_{k,l}&=-\eta_{k,l}w^{(1)}_{k,l}-i\gamma kU\frac{2F-\lambda_{k,l}}{2F+\lambda_{k,l}}w^{(2)}_{k,l}, \\
Lw^{(2)}_{k,l}&=-\frac{\lambda_{k,l}}{2F+\lambda_{k,l}}\eta_{k,l} w^{(2)}_{k,l}+ i\gamma kU w^{(1)}_{k,l}. 
\end{aligned}
\end{equation}

We then extract from \eqref{Lw} the coefficient matrix
\begin{equation}\label{Mkl}
M_{k,l} = \left(\begin{matrix}
-\eta_{k,l} & -i\gamma\frac{2F-\lambda_{k,l}}{2F+\lambda_{k,l}}kU \\
i\gamma kU & -\frac{\lambda_{k,l}}{2F+\lambda_{k,l}}\eta_{k,l}
\end{matrix}\right).
\end{equation}

It is easy to see that if $V\in \C^2$ is a left eigenvector of $M_{k,l}$ then $w=w^{(1)}_{k,l} V_1+w^{(2)}_{k,l} V_2$ is an eigenfunction of $L$ with the same eigenvalue.

Now, the left eigenvectors of $M_{k,l}$ are
\begin{equation}\label{lefteigenvec}
V_{k,l}^\pm = \left( i\gamma kU , \mu^\pm_{k,l}+\eta_{k,l} \right),
\end{equation}
where the corresponding eigenvalues $\mu^\pm_{k,l}$ are given by \eqref{mupm}.

The claim about $\sigma(L)$ follows from the above, and we also see that  
$$w^{\pm}_{k,l}= i\gamma kU w^{(1)}_{k,l}+(\mu^\pm_{k,l}+\eta_{k,l})w^{(2)}_{k,l}$$ 
are eigenfunctions of $L$, in the complexification of $H$, with corresponding eigenvalues $\mu^{\pm}_{k,l}$.

The spaces $H_{k,l,\pm}$ are then constructed by taking the span of the real and imaginary parts of $w^\pm_{k,l}$. These functions can be easily seen to take the form given in \eqref{Hkl+-}, and the rest of the statements follows now directly  from the definition of $w^\pm_{k,l}$.
\end{proof}

Having a complete quantitative description of the spectrum of $L$, we now turn to the problem of describing how linear instability can occur in terms of the parameters Froude number and $F$ and the shear velocity $U$ of the basic flow.

\begin{lem}\label{PES}
Let
\begin{equation}\label{Z0}
\mathcal{Z}_0 = \left\lbrace (k,l)\in \N^2:\: U_c^2 = \mathcal{U}^2(k,l) \right\rbrace.
\end{equation}
Then the following principle of exchange of stability (PES) holds true 
in terms of $F$ and $U$:

\begin{enumerate}
\item If $2F\leq \gamma^2+1$, then $\Re \mu^\pm_{k,l}<0$ for all $k\geq 0$ and $l\geq 1$,

\item If $2F>\gamma^2+1$, then there exists a neighborhood $\mathcal{N}_c$ of $U_c$ such that, for all $U\in\mathcal{N}_c$, we have
\begin{equation}\label{PES_eq}
\begin{aligned}
&\Re \mu^+_{k,l} \begin{cases}
<0 \quad \text{if } U<U_c, \\
=0 \quad \text{if } U=U_c, \\
>0 \quad \text{if } U>U_c, \\
\end{cases}  &&  \forall (k,l)\in \mathcal{Z}_0,\\
&\Re \mu^-_{k,l}<0 && \forall (k,l)\in \mathcal{Z}_0,\\
&\Re \mu^\pm_{k,l} < 0&& \forall (k,l)\not\in \mathcal{Z}_0.
\end{aligned}
\end{equation}

\end{enumerate}

\end{lem} 

\begin{proof}

First we note that $\Re \mu^-_{k,l} \leq -\frac{(F+\lambda_{k,l})\Re\eta_{k,l}}{2F+\lambda_{k,l}} <0$ for all $k,l$, so this part of the spectrum of $L$ is always stable, irrespective of the value of $U$.

Next we fix an integer pair $(k,l)$, and assume that there exists a $\bar{U}\in \R$ for which $\Re\mu^+_{k,l}(\bar{U})= 0$. We claim that this is only possible if $2F>\lambda_{k,l}$. Note that this already proves the lemma in the case $2F\leq \gamma^2+1$.

To show the claim, let $\mu^+_{k,l}$ be of the form $\mu^+_{k,l}=i\sigma$ for some $\sigma\in\R$. Since the characteristic polynomial of $M_{k,l}$ is given by
\begin{equation*}
z^2 + \frac{2(F+\lambda_{k,l})\eta_{k,l}}{2F+\lambda_{k,l}}z + \frac{\lambda_{k,l}\eta_{k,l}^2-\gamma^2k^2U^2(2F-\lambda_{k,l})}{2F+\lambda_{k,l}}, 
\end{equation*}
we see that $\sigma$ and $\bar{U}$ must satisfy
\begin{align*}
& (2F+\lambda_{k,l})\sigma^2 +2(F+\lambda_{k,l})\Im \eta_{k,l}\sigma -\lambda_{k,l}\Re(\eta_{k,l}^2)+\gamma^2k^2\bar{U}^2(2F-\lambda_{k,l})=0,  \\
& (F+\lambda_{k,l})\sigma+ \lambda_{k,l}\Im\eta_{k,l}=0.
\end{align*}

It is easy to see that the above system has a solution if and only if $2F>\lambda_{k,l}$, whence the claim follows, and, when that is the case, the values of $\bar{U}$ and $\sigma$ are uniquely determined by $(k,l)$ as given below:
\begin{equation}\label{U2}
\begin{aligned}
&\bar{U}^2 = \frac{1}{2F-\lambda_{k,l}} \left( \frac{F^2\beta^2}{\lambda_{k,l}(F+\lambda_{k,l})^2}+\frac{\lambda_{k,l}(\frac{1}{Re}\lambda_{k,l}+r)^2}{\gamma^2k^2} \right), \\
&\sigma = -\frac{\gamma k\beta}{F+\lambda_{k,l}}.
\end{aligned}
\end{equation}

Suppose now that $2F>\lambda_{k,l}$ for a fixed pair $(k,l)$. Then we have
\begin{equation*}
\Re[(2F+\lambda_{k,l})\mu^+(0)+(F+\lambda_{k,l})\eta_{k,l}] = F\Re \eta_{k,l}>0,
\end{equation*}
which, together with the identity
\begin{equation*}
\frac{\partial}{\partial U} \Re \mu^+_{k,l}= \frac{\gamma^2 k^2 (2F-\lambda_{k,l})U\Re[(2F+\lambda_{k,l})\mu^+_{k,l}+(F+\lambda_{k,l})\eta_{k,l}]}{|(2F+\lambda_{k,l})\mu^+_{k,l}+(F+\lambda_{k,l})\eta_{k,l}|^2},
\end{equation*}
shows that the map $U\mapsto \Re \mu^+_{k,l}(U)$ is strictly increasing for $U>0$.

The rest of the lemma follows immediately from the definition of $U_c^2$, the monotonicity of $\Re\mu^+_{k,l}$ with respect to $U$, and the fact that $\Re\mu^+_{k,l}|_{U=0}<0$ for all $k,l$.

\end{proof}

\section{Proof of Theorem~\ref{transition}}

Theorem \ref{transition} follows essentially from the theory of dynamic transitions from simple complex eigenvalues laid out in \cite{ptd}. To make use of these ideas we need the following result, which provides a precise description of the local behavior of the system near the equilibrium.

\begin{thm}\label{reduced}
Assume $\mathcal{Z}_0=\{(\hat{k},\hat{l})\}$, and let $H_c=H_{\hat{k},\hat{l},+}$ be as defined in \eqref{Hkl+-}. Then \eqref{abst} has a local invariant manifold of the form 
$$\mathcal{M} = \{ x+h(x):\: x\in H_c,\, |x|\leq r_0\},$$
where $r_0>0$  is independent of $U$ and
\begin{equation}\label{mnfld}
\begin{aligned}
& h(x)= -a(x_1^2+x_2^2)e^{(4)}_{0,2\hat{l}}+ O(|x|^3),\\
&  x=x_1 v^{1,+}_{\hat{k},\hat{l}}+x_2 v^{2,+}_{\hat{k},\hat{l}},\\
& a=\frac{\gamma^2 \hat{k}^2 \hat{l} U F\Re(\mu^+_{\hat{k},\hat{l}}+\eta_{\hat{k},\hat{l}})}{(2F+\lambda_{0,2\hat{l}})(2\alpha^+_{\hat{k},\hat{l}}-\alpha^+_{0,2\hat{l}})}.
\end{aligned}
\end{equation}
 
Furthermore, the flow induced by \eqref{abst} on $\mathcal{M}$ is governed by the equations
\begin{equation}\label{red_eq}
\begin{aligned}
\frac{dx_1}{dt}= \alpha x_1-\sigma x_2-A(x_1^2+x_2^2)(\cos\phi\, x_1-\sin\phi\,  x_2)+O(|x|^4),\\
\frac{dx_2}{dt}= \alpha x_2+\sigma x_1-A(x_1^2+x_2^2)(\sin\phi\, x_1+\cos\phi\, x_2)+O(|x|^4),
\end{aligned}
\end{equation}
where  
\begin{equation}
\label{A const}
\begin{aligned}
& \alpha=\alpha^+_{\hat{k},\hat{l}}, \qquad \sigma=-\sigma^+_{\hat{k},\qquad  \hat{l}}, e^{i\phi}=\frac{\mu^+_{\hat{k},\hat{l}}-\mu^-_{\hat{k},\hat{l}}}{|\mu^+_{\hat{k},\hat{l}}-\mu^-_{\hat{k},\hat{l}}|}, \\
& 
A=\frac{\gamma^4 \hat{k}^4\hat{l}^2 U^2 F \Re(\mu^+_{\hat{k},\hat{l}}+\eta_{\hat{k},\hat{l}})[(2F-\lambda_{\hat{k},\hat{l}})(2\lambda_{\hat{k},\hat{l}}-\lambda_{0,2\hat{l}})+\lambda_{0,2\hat{l}}\lambda_{\hat{k},\hat{l}}]}{\lambda_{\hat{k},\hat{l}}(2F+\lambda_{\hat{k},\hat{l}})(2F+\lambda_{0,2\hat{l}})(2\alpha^+_{\hat{k},\hat{l}}-\alpha^+_{0,2\hat{l}})|\mu^+_{\hat{k},\hat{l}}-\mu^-_{\hat{k},\hat{l}}|}.
\end{aligned}
\end{equation}

\end{thm}

\begin{rem}
The invariant manifold given above is known to be asymptotically complete; see \cite{ptd}, so the behavior of the full system \eqref{main1} near the origin is the same as that of \eqref{red_eq}. Theorem \ref{transition} is thus a consequence of Theorem \ref{reduced}, as we show below. 
\end{rem}

\begin{proof}[Proof of Theorem \ref{transition}]
Note that $A\cos\phi >0$ if and only if \eqref{cont_condition} is satisfied, and in that case the origin is a stable equilibrium of \eqref{red_eq} when $\alpha=0$. This fact already shows that the transition is continuous, see \cite{ptd}. Hence it only remains to obtain the approximation formulas \eqref{approx_soln}.

By changing to polar coordinates, and assuming \eqref{cont_condition} is satisfied, it is easy to see that the system \eqref{red_eq} has a periodic solution given by
\begin{equation}\label{x1 x2}
x_1(t) = \rho\cos(\omega t),\quad x_2(t)= \rho\sin(\omega t),
\end{equation}
where
\begin{equation}\label{rho omega}
\rho^2 = \frac{\alpha}{A\cos\phi} +O\left(\alpha^{{3}/{2}}\right),\qquad \omega = \sigma-\alpha\tan\phi + O\left(\alpha^{{3}/{2}}\right).
\end{equation}

The corresponding solution to \eqref{abst} is thus given by
\begin{equation*}
\begin{aligned}
u(t)&= \rho \cos (\omega t) v^{1,+}_{k,l}+\rho \sin(\omega t) v^{2,+}_{k,l} +h\left(\rho \cos (\omega t) v^{1,+}_{k,l}+\rho \sin(\omega t) v^{2,+}_{k,l} \right) \\
&= \left(\frac{\alpha}{A\cos \phi}\right)^{{1}/{2}} \left(\cos(\omega t)v^{1,+}_{k,l}+\sin(\omega t) v^{2,+}_{k,l}\right)+O\left(\alpha \right).
\end{aligned}
\end{equation*}
Here we have dropped the hats over $k$  and $l$ to simplify the exposition.

Using \eqref{Hc}, \eqref{lefteigenvec} and \eqref{wi_kl} we find that
\begin{equation*}
\begin{aligned}
\cos(\omega t)v^{1,+}_{k,l}+\sin(\omega t) v^{2,+}_{k,l} &= \Re \left[ e^{-i\omega t} w^+_{k,l} \right] \\
&= \Re\left[ e^{-i\omega t} \left( i\gamma k U w^{(1)}_{k,l} +(\mu^+_{k,l}+\eta_{k,l}) w^{(2)}_{k,l} \right)\right].
\end{aligned}
\end{equation*}
Then by definition,  we have 
\begin{equation*}
e^{-i\omega t} \left( i\gamma k U w^{(1)}_{k,l} +(\mu^+_{k,l}+\eta_{k,l}) w^{(2)}_{k,l} \right)
= \begin{pmatrix} i\gamma k U e^{-i(\omega t +\gamma k x)} \sin (ly) \\
(\mu^+_{k,l}+\eta_{k,l})e^{-i(\omega t +\gamma k x)} \sin (ly) \end{pmatrix}.
\end{equation*}
Consequently, 
\begin{equation*}
\cos(\omega t)v^{1,+}_{k,l}+\sin(\omega t) v^{2,+}_{k,l}  
= \begin{pmatrix} \gamma k U \sin(\omega t +\gamma k x) \sin (ly) \\
\Re [(\mu^+_{k,l}+\eta_{k,l})e^{-i(\omega t +\gamma k x)}] \sin (ly) \end{pmatrix},
\end{equation*}
with
\begin{equation*}
\begin{aligned}
\Re [(\mu^+_{k,l}+\eta_{k,l})e^{-i(\omega t +\gamma k x)}] =& \left(\frac{1}{Re}\lambda_{k,l}+r\right)\cos(\omega t + \gamma k x)\\
& +\dfrac{\gamma k \beta F}{\lambda_{k,l} (F+\lambda_{k,l})}\sin(\omega t + \gamma k x) + O(\alpha).
\end{aligned}
\end{equation*}

Thus we obtain \eqref{approx_soln} by making use of the above calculations and \eqref{rho omega}.

To show \eqref{rho omega thm} we note first that at $U=U_c$, one has
\begin{align*}
& \frac{d\mu^+_{k,l}}{dU}=\frac{\gamma^2k^2(2F-\lambda_{k,l}) (F+\lambda_{k,l})U_c}{(F+\lambda_{k,l})^4(\Re\eta_{k,l})^2+F^4(\Im\eta_{k,l})^2}\left( (F+\lambda_{k,l})^2\Re\eta_{k,l}-iF^2\Im\eta_{k,l}\right), \\
& \mu^+_{k,l}-\mu^-_{k,l} = \frac{2}{(2F+\lambda_{k,l})(F+\lambda_{k,l})} \left((F+\lambda_{k,l})^2\Re\eta_{k,l}+iF^2\Im\eta_{k,l}\right),
\end{align*}
so that, upon substitution in \eqref{A const}, we obtain
\begin{equation*}
\frac{1}{A\cos\phi} =\frac{\lambda_{k,l}\lambda_{0,2l}\Re\eta_{0,2l}(2F+\lambda_{k,l})|\mu^+_{k,l}-\mu^-_{k,l}|^2}{\gamma^4 k^4l^2 U^2 F \Re\eta_{k,l}[(2F-\lambda_{k,l})(2\lambda_{k,l}-\lambda_{0,2l})+\lambda_{0,2l}\lambda_{k,l}]\Re(\mu^+_{k,l}-\mu^-_{k,l})}.
\end{equation*}

Next we approximate $\alpha= \frac{d\mu^+_{k,l}}{dU}\Big|_{U=U_c}(U-U_c)+O(|U-U_c|^2)$ to get 
\begin{equation*}
\frac{\alpha}{A\cos\phi} =\frac{4\lambda_{k,l}(2F-\lambda_{k,l})\Re\eta_{0,2l}(U-U_c)}{\gamma^2 k^2U_c F [(2F-\lambda_{k,l})(2\lambda_{k,l}-\lambda_{0,2l})+\lambda_{0,2l}\lambda_{k,l}]\Re\eta_{k,l}}+O(|U-U_c|^2).
\end{equation*}
which proves the formula for $\rho$ in \eqref{rho omega thm} after some straightforward simplifications. Note that the $O(|U-U_c|^2)$  are absorbed into the $O(|U-U_c|)$ terms in 
\eqref{approx_soln}.

Using the same arguments we find that
\[\frac{d}{dU}\Big|_{U=U_c} (\sigma-\alpha\tan \phi ) = -\frac{d}{dU}\Big|_{U=U_c}\Im\mu^+_{k,l} - \tan\phi \frac{d}{dt}\Big|_{U=U_c}\Re\mu^+_{k,l} =0,
\]
which shows that the $O(\alpha)=O(|U-U_c|)$ terms in the formula for $\omega$ vanish identically, as desired.
\end{proof}

The proof of Theorem \ref{reduced} relies on the explicit form of the eigenfunctions of $L$ and their non-linear interactions. In what follows we describe in detail the process needed to approximate the invariant manifold function.

\begin{proof}[Proof of Theorem \ref{reduced}]

The existence of a local invariant manifold, tangent to the center-unstable subspace $H_c$ at the origin, is known to hold by, e.g, \cite{ptd}.

Hence we proceed directly to the approximation of the invariant manifold function noting that, although the following is a purely formal approach, the formulas we obtain do have a rigorous justification.

Using the ansatz $u(t)= u_c(t)+h(u_c(t))$, where $u_c=P_c u$ and $h:H_c\to H_s$ is to be determined, we see that \eqref{abst} can be satisfied only when
\begin{equation}\label{h}
\begin{aligned}
&Dh(x)\left[ Lx+P_c\left(G(x,x)+\tilde{G}(x,h(x))+G(h(x),h(x)) \right)\right] \\
& \qquad = Lh(x)+P_s \left[ G(x,x)+\tilde{G}(x,h(x))+G(h(x),h(x)) \right]\qquad\forall x\in H_c,
\end{aligned}
\end{equation}
where by abuse of notation we denote by $G$ both the non-linearity in \eqref{abst} and the bilinear form such that $G(u)=G(u,u)$. Note also that we write $\tilde{G}$ for the symmetric part of $G$, $\tilde{G}(x,y)=G(x,y)+G(y,x)$.

We next proceed to compute the non-linear interactions. After some calculations we find that the only non-zero elements arising from the evaluation of this bilinear form on the basis elements of $H_{k,l}$ are
\begin{equation}\label{BHkl}
\begin{aligned}
G(e^{(1)}_{k,l} , e^{(2)}_{k,l} ) &= -\frac{\gamma kl\lambda_{k,l}}{2\lambda_{0,2l}} e^{(2)}_{0,2l},\quad && G(e^{(2)}_{k,l} , e^{(1)}_{k,l} ) = \frac{\gamma kl\lambda_{k,l}}{2\lambda_{0,2l}} e^{(2)}_{0,2l}, \\
G(e^{(1)}_{k,l} , e^{(4)}_{k,l} ) &= -\frac{\gamma kl (2F+\lambda_{k,l})}{2(2F+\lambda_{0,2l})} e^{(4)}_{0,2l},\quad && G(e^{(4)}_{k,l} , e^{(1)}_{k,l} ) =  \frac{\gamma kl \lambda_{k,l} }{2(2F+\lambda_{0,2l})} e^{(4)}_{0,2l},\\
G(e^{(2)}_{k,l} , e^{(3)}_{k,l} ) &= \frac{\gamma kl (2F+\lambda_{k,l})}{2(2F+\lambda_{0,2l})} e^{(4)}_{0,2l},\quad && G(e^{(3)}_{k,l} , e^{(2)}_{k,l} ) = -\frac{\gamma kl \lambda_{k,l} }{2(2F+\lambda_{0,2l})} e^{(4)}_{0,2l},\\
G(e^{(3)}_{k,l} , e^{(4)}_{k,l} ) &= -\frac{\gamma kl\lambda_{k,l}}{2\lambda_{0,2l}} e^{(2)}_{0,2l},\quad && G(e^{(4)}_{k,l} , e^{(3)}_{k,l} ) = \frac{\gamma kl\lambda_{k,l}}{2\lambda_{0,2l}} e^{(2)}_{0,2l}.
\end{aligned}
\end{equation}

The above readily shows that for $x\in H_c\subset H_{k,l}$ one has $G(x,x)\in \text{span }\{ e^{(4)}_{0,2l}\}$. This suggests that we seek for $h$ in the form
\[h(x)= (\tilde{a}_1 x_1^2+\tilde{a}_2 x_1x_2+\tilde{a}_3 x_2^2) e^{(4)}_{0,2l}+O(|x|^3),\]
where the coefficients $\tilde{a}_i$, $i=1,2,3$, are to be determined.

Next we need a more explicit formula for $G(x,x)$ when $x$ is given in the form $x=x_1v^{1,+}_{k,l}+x_2 v^{2,+}_{k,l}$. This is accomplished by making use of \eqref{LHkl} and \eqref{BHkl}, from which  we derive
\begin{equation}\label{BHc}
\begin{aligned}
&G(v^{1,\pm}_{k,l} , v^{1,\pm}_{k,l})= G(v^{2,\pm}_{k,l} , v^{2,\pm}_{k,l}) = -\frac{\gamma^2 k^2lUF\Re(\mu^\pm_{k,l}+\eta_{k,l})}{2F+\lambda_{0,2l}} e^{(4)}_{0,2l},\\
&G(v^{1,\pm}_{k,l},v^{2,\pm}_{k,l}) + G(v^{2,\pm}_{k,l}, v^{1,\pm}_{k,l}) =0.
\end{aligned}
\end{equation}

Then, using \eqref{BHc} and substituting the ansatz for $h$ in \eqref{h}, it is easy to see that $\tilde{a}_1=\tilde{a}_3$ and $\tilde{a}_2=0$. Indeed, after some calculations we obtain
\begin{equation*}
h(x)= -\frac{\gamma^2 k^2lUF\Re(\mu^+_{k,l}+\eta_{k,l})(x_1^2+x_2^2)}{(2F+\lambda_{0,2l})\Re(2\mu^+_{k,l}-\mu^+_{0,2l})} e^{(4)}_{0,2l}+ O(|x|^3).
\end{equation*}

Hence, with $a$ given by \eqref{mnfld}, the reduced equations take the form 
\begin{equation}\label{red0}
\begin{aligned}
\frac{dx_1}{dt}=\alpha^+_{k,l}x_1+\sigma^+_{k,l}x_2 -a(x_1^2+x_2^2) P_1 \left[ x_1\tilde{G}(v^{1,+}_{k,l},e^{(4)}_{0,2l})+x_2\tilde{G}(v^{2,+}_{k,l},e^{(4)}_{0,2l}) \right], \\
\frac{dx_2}{dt}=\alpha^+_{k,l}x_2-\sigma^+_{k,l}x_1 -a(x_1^2+x_2^2) P_2 \left[ x_1\tilde{G}(v^{1,+}_{k,l},e^{(4)}_{0,2l})+x_2\tilde{G}(v^{2,+}_{k,l},e^{(4)}_{0,2l}) \right], 
\end{aligned}
\end{equation}
where $P_i$ is the projection onto the span of $v^{i,+}_{k,l}$.

The terms within brackets in \eqref{red0} require further calculations. To begin with, we need to compute the non-linear interactions between elements in $H_{k,l}$ and $H_{0,2l}$. The results of such straightforward, albeit lengthy, calculations are as follows:
\begin{equation}\label{BHklH02l}
\begin{aligned}
&\tilde{G}(e^{(1)}_{k,l},e^{(4)}_{0,2l}) = \gamma kl(2F+\lambda_{0,2l}-\lambda_{k,l})\left( \frac{1}{2F+\lambda_{k,l}} e^{(4)}_{k,l}-\frac{1}{2F+\lambda_{k,3l}} e^{(4)}_{k,3l}\right), \\
&\tilde{G}(e^{(2)}_{k,l},e^{(4)}_{0,2l}) = -\gamma kl(2F+\lambda_{0,2l}-\lambda_{k,l})\left(\frac{1}{2F+\lambda_{k,l}} e^{(3)}_{k,l}-\frac{1}{2F+\lambda_{k,3l}} e^{(3)}_{k,3l}\right), \\
&\tilde{G}(e^{(3)}_{k,l},e^{(4)}_{0,2l}) = \gamma kl(\lambda_{0,2l}-\lambda_{k,l})\left(\frac{1}{\lambda_{k,l}} e^{(2)}_{k,l}-\frac{1}{\lambda_{k,3l}} e^{(2)}_{k,3l}\right), \\
&\tilde{G}(e^{(4)}_{k,l},e^{(4)}_{0,2l}) = -\gamma kl(\lambda_{0,2l}-\lambda_{k,l})\left(\frac{1}{\lambda_{k,l}} e^{(1)}_{k,l}-\frac{1}{\lambda_{k,3l}} e^{(1)}_{k,3l}\right). 
\end{aligned}
\end{equation}

Now, since we are given $x$ in terms of the basis $\{v^{1,+}_{k,l}, v^{2,+}_{k,l}\}$, we need to substitute \eqref{Hc} in \eqref{BHklH02l}, and we obtain that 
\begin{align}\label{Bv1_e02l}
& \tilde{G}(v^{1,+}_{k,l}, e^{(4)}_{0,2l}) \\
& \quad  \nonumber = \gamma^2 k^2lU(2F+\lambda_{0,2l}-\lambda_{k,l})\left( \frac{1}{2F+\lambda_{k,l}} e^{(4)}_{k,l} - \frac{1}{2F+\lambda_{k,3l}} e^{(4)}_{k,3l}\right) \\
&\qquad  \nonumber +\frac{\gamma kl(\lambda_{0,2l}-\lambda_{k,l})}{\lambda_{k,l}}\left( -\Re(\mu^+_{k,l}+\eta_{k,l}) e^{(1)}_{k,l}+\Im(\mu^+_{k,l}+\eta_{k,l})e^{(2)}_{k,l} \right) \\
& \qquad \nonumber -\frac{\gamma kl(\lambda_{0,2l}-\lambda_{k,l})}{\lambda_{k,3l}}\left( -\Re(\mu^+_{k,l}+\eta_{k,l}) e^{(1)}_{k,3l}+\Im(\mu^+_{k,l}+\eta_{k,l})e^{(2)}_{k,3l} \right),
\\
& \label{Bv2_e02l}
\tilde{G}(v^{2,+}_{k,l}, e^{(4)}_{0,2l}) \\
& \quad \nonumber = -\gamma^2 k^2lU(2F+\lambda_{0,2l}-\lambda_{k,l})\left( \frac{1}{2F+\lambda_{k,l}} e^{(3)}_{k,l} - \frac{1}{2F+\lambda_{k,3l}} e^{(3)}_{k,3l}\right) \\
&\qquad \nonumber  -\frac{\gamma kl(\lambda_{0,2l}-\lambda_{k,l})}{\lambda_{k,l}}\left( \Im(\mu^+_{k,l}+\eta_{k,l}) e^{(1)}_{k,l}+\Re(\mu^+_{k,l}+\eta_{k,l})e^{(2)}_{k,l} \right) \\
&\qquad \nonumber +\frac{\gamma kl(\lambda_{0,2l}-\lambda_{k,l})}{\lambda_{k,3l}}\left( \Im(\mu^+_{k,l}+\eta_{k,l}) e^{(1)}_{k,3l}+\Re(\mu^+_{k,l}+\eta_{k,l})e^{(2)}_{k,3l} \right).
\end{align}

In order to project the above onto $H_{k,l,+}$ we use the identities:
\begin{equation}
\begin{aligned}
-\Re(\mu^+_{k,l}+\eta_{k,l}) e^{(1)}_{k,l}+\Im(\mu^+_{k,l}+\eta_{k,l})e^{(2)}_{k,l} = \Im\left[ (\mu^+_{k,l}+\eta_{k,l})(e^{(2)}_{k,l}-i e^{(1)}_{k,l})\right], \\
\Im(\mu^+_{k,l}+\eta_{k,l}) e^{(1)}_{k,l}+\Re(\mu^+_{k,l}+\eta_{k,l})e^{(2)}_{k,l} = \Re\left[ (\mu^+_{k,l}+\eta_{k,l})(e^{(2)}_{k,l}-i e^{(1)}_{k,l})\right],
\end{aligned}
\end{equation}
and
\begin{equation}\label{w_to_e}
w^\pm_{k,l}=v^{1,\pm}_{k,l}+i v^{2,\pm}_{k,l} = i\gamma kU (e^{(2)}_{k,l}-i e^{(1)}_{k,l}) + (\mu^\pm_{k,l}+\eta_{k,l}) (e^{(4)}_{k,l}-i e^{(3)}_{k,l}).
\end{equation}

Thus, choosing the minus sign in \eqref{w_to_e} and multiplying by $\mu^+_{k,l}+\eta_{k,l}$, we see that
\begin{equation*}
P_c \left[ (\mu^+_{k,l}+\eta_{k,l}))(e^{(2)}_{k,l}-i e^{(1)}_{k,l})\right] = -\frac{(\mu^+_{k,l}+\eta_{k,l})(\mu^-_{k,l}+\eta_{k,l})}{i\gamma kU} P_c\left(e^{(4)}_{k,l}-i e^{(3)}_{k,l} \right).
\end{equation*}

Next, using the fact that
\begin{equation*}
\frac{(\mu^+_{k,l}+\eta_{k,l})(\mu^-_{k,l}+\eta_{k,l})}{i\gamma kU} = \frac{i\gamma kU(2F-\lambda_{k,l})}{2F+\lambda_{k,l}},
\end{equation*}
we obtain
\begin{align*}
& P_c \left[ -\Re(\mu^+_{k,l}+\eta_{k,l}) e^{(1)}_{k,l}+\Im(\mu^+_{k,l}+\eta_{k,l})e^{(2)}_{k,l} \right] = -\frac{\gamma kU(2F-\lambda_{k,l})}{2F+\lambda_{k,l}} P_c e^{(4)}_{k,l}, \\
& P_c \left[ \Im(\mu^+_{k,l}+\eta_{k,l}) e^{(1)}_{k,l}+\Re(\mu^+_{k,l}+\eta_{k,l})e^{(2)}_{k,l} \right] = -\frac{\gamma kU(2F-\lambda_{k,l})}{2F+\lambda_{k,l}} P_c e^{(3)}_{k,l}.
\end{align*}

With the aid of the above formulae are now in position to compute the projections of \eqref{Bv1_e02l} and \eqref{Bv2_e02l} onto $H_{k,l}$, which take the form
\begin{align}
\label{PcBv1_e02l}
& P_c\tilde{G}(v^{1,+}_{k,l}, e^{(4)}_{0,2l}) \\
& \qquad \nonumber = \frac{\gamma^2 k^2lU}{(2F+\lambda_{k,l})\lambda_{k,l}} \left[(2F-\lambda_{k,l})(2\lambda_{k,l}-\lambda_{0,2l})+\lambda_{0,2l}\lambda_{k,l} \right] P_c e^{(4)}_{k,l}, 
\\
& \label{PcBv2_e02l}
P_c\tilde{G}(v^{2,+}_{k,l}, e^{(4)}_{0,2l}) \\
& \qquad \nonumber = -\frac{\gamma^2 k^2lU}{(2F+\lambda_{k,l})\lambda_{k,l}} \left[(2F-\lambda_{k,l})(2\lambda_{k,l}-\lambda_{0,2l})+\lambda_{0,2l}\lambda_{k,l} \right] P_c e^{(3)}_{k,l}.
\end{align}

We need to take an additional projection of these equations to get their components in $H_{k,l,+}$. For this purpose we first subtract the different signs of \eqref{w_to_e} to get 
\begin{equation*}
(\mu^+_{k,l}-\mu^-_{k,l})(e^{(4)}_{k,l}-i e^{(3)}_{k,l}) = (v^{1,+}_{k,l}-v^{1,-}_{k,l}) + i ( v^{2,+}_{k,l}-v^{2,-}_{k,l}), 
\end{equation*}
which implies that
\begin{equation}\label{Pce4e3}
\begin{aligned}
&P_c e^{(4)}_{k,l} =\Re \left(\frac{1}{\mu^+_{k,l}-\mu^-_{k,l}}\right)v^{1,+}_{k,l}-\Im \left(\frac{1}{\mu^+_{k,l}-\mu^-_{k,l}}\right) v^{2,+}_{k,l}, \\
&P_c e^{(3)}_{k,l} = -\Im \left(\frac{1}{\mu^+_{k,l}-\mu^-_{k,l}}\right) v^{1,+}_{k,l}-\Re \left(\frac{1}{\mu^+_{k,l}-\mu^-_{k,l}}\right) v^{2,+}_{k,l}.
\end{aligned}
\end{equation}

Finally, upon substituting \eqref{Pce4e3} in \eqref{PcBv1_e02l} and \eqref{PcBv2_e02l}, we get the projections onto $H_{k,l,+}$ as follows
\begin{equation}\label{final}
\begin{aligned}
&P_1 \tilde{G}(v^{1,+}_{k,l}, e^{(4)}_{k,l})=P_2 \tilde{G}(v^{2,+}_{k,l}, e^{(4)}_{k,l})=a_{kl}  \Re \left(\frac{1}{\mu^+_{k,l}-\mu^-_{k,l}}\right),\\
&P_2 \tilde{G}(v^{1,+}_{k,l}, e^{(4)}_{k,l})=-P_1 \tilde{G}(v^{2,+}_{k,l}, e^{(4)}_{k,l})=-a_{kl}  \Im \left(\frac{1}{\mu^+_{k,l}-\mu^-_{k,l}}\right),
\end{aligned}
\end{equation}
where
\begin{equation*}
a_{kl}=\frac{\gamma^2 k^2lU\left[(2F-\lambda_{k,l})(2\lambda_{k,l}-\lambda_{0,2l})+\lambda_{0,2l}\lambda_{k,l} \right]}{(2F+\lambda_{k,l})\lambda_{k,l}}.
\end{equation*}

The final result follows upon substitution of \eqref{final} in \eqref{red0}.

\end{proof}

\bibliographystyle{siam}

\end{document}